\pdfoutput=1
\documentclass[11pt]{article}
\usepackage{subfig}
\usepackage{amssymb,amsmath,amsthm}
\usepackage{graphics,epsfig}
\usepackage{algorithmic}
\usepackage{algorithm}
\usepackage[margin=1 in]{geometry}

\newtheorem{lem}{Lemma}[section]

\newtheorem{thm}[lem]{Theorem}

\title{Scheduling Coflows for Minimizing the Total Weighted Completion Time in Identical Parallel Networks}
\author{Chi-Yeh~Chen 
\\ Department of Computer Science and Information
Engineering, \\ National Cheng Kung University, \\
Taiwan, ROC. \\
chency@csie.ncku.edu.tw.}

\begin{document}

\maketitle
\begin{abstract}
Coflow is a recently proposed network abstraction to capture communication patterns in data centers. The coflow scheduling problem in large data centers is one of the most important $NP$-hard problems. Previous research on coflow scheduling focused mainly on the single-switch model. However, with recent technological developments, this single-core model is no longer sufficient. This paper considers the coflow scheduling problem in identical parallel networks. The identical parallel network is an architecture based on multiple network cores running in parallel. Coflow can be considered as divisible or indivisible. Different flows in a divisible coflow can be transmitted through different network cores. Considering the divisible coflow scheduling problem, we propose a $(6-\frac{2}{m})$-approximation algorithm with arbitrary release times, and a $(5-\frac{2}{m})$-approximation without release time, where $m$ is the number of network cores. On the other hand, when coflow is indivisible, we propose a $(4m+1)$-approximation algorithm with arbitrary release times, and a $(4m)$-approximation without release time.

\end{abstract}

\section{Introduction}\label{sec:Introduction}
Over the past decade, large data centers have become the dominant form of computing infrastructure. Numerous studies~\cite{Chowdhury2014, Chowdhury2015, Zhang2016, Agarwal2018} have demonstrated the benefits of application-aware network scheduling by exploiting the structured traffic patterns of distributed applications in data centers. Data-parallel computing applications such as MapReduce~\cite{Dean2008}, Hadoop~\cite{Shvachko2010}, and Spark~\cite{zaharia2010spark} consist of multiple stages of computation and communication. The success of these applications has led to a proliferation of applications designed to alternate between computational and communication stages. In data-parallel computing applications, computation involves only local operations in the server. However, much of the intermediate data (flows) generated during the computation stage need to be transmitted across different machines during the communication stage for further processing. As the number of applications increases, data centers require more data transfer capability. In these data transfers, the collective impact of all flows between the two machine groups becomes important. This collective communication pattern in the data center is abstracted by coflow traffic~\cite{Chowdhury2012}.

A coflow represents a collection of related flows whose completion time is determined by the completion time of the last flow in the collection~\cite{shafiee2018improved}. A data center can be modeled as giant $N \times N$ non-blocking switch, with $N$ input links connected to $N$ source servers and $N$ output links connected to $N$ destination servers. Moreover, we assume that the transmission speed of each port is uniform. This modeling allows us to focus only on scheduling tasks, rather than routing flows. Each coflow can be represented by a $N\times N$ integer matrix $D=\left(d_{ij}\right)_{i,j=1}^{N}$, where entry $d_{ij}$ represents the number of data units that must be transferred from input port $i$ to output port $j$. Each coflow also has a weight and a release time. Weightings can capture different priorities for different coflows.

Previous research on coflow scheduling focused mainly on the single-switch model, which has been widely used in coflow studies. However, with recent technological developments, this single-core model is no longer sufficient. In fact, a growing data center would operate on multiple networks in parallel to increase the efficiency of the network~\cite{Singh2015, Huang2020}. We consider the identical parallel network, an architecture based on multiple network cores running in parallel. Parallel networks provide a large amount of aggregated bandwidth by serving traffic at the same time. The goal of this paper is to schedule coflows in the identical parallel networks such that the weighted completion time is minimized. Coflow can be considered as divisible or indivisible. Different flows in a divisible coflow can be transmitted through different network cores. We assume that divisible coflows are transmitted at the flow level, so that data in a flow are distributed to the same network core. On the other hand, flows in an indivisible coflow can only be transmitted through the same network core.

\subsection{Related Work}
Chowdhury and Stoica~\cite{Chowdhury2012} first introduced the coflow abstraction to capture communication patterns in data centers. Since then, many related studies have been proposed in the literature to schedule coflows, e.g. \cite{Chowdhury2014, Chowdhury2015, Qiu2015, zhao2015rapier, shafiee2018improved, ahmadi2020scheduling}. 
It is well known that the concurrent open shop problem is NP-hard to approximate within a factor better than $2-\epsilon$ for any $\epsilon>0$~\cite{Sachdeva2013, shafiee2018improved}. Since the coflow scheduling problem generalizes the well-studied concurrent open shop scheduling problem, it is also NP-hard to approximate within a factor better than $2-\epsilon$~\cite{ahmadi2020scheduling, Bansal2010, Sachdeva2013}. Qiu \textit{et al.}~\cite{Qiu2015} developed deterministic $\frac{64}{3}$ approximation and randomized $(8+\frac{16\sqrt{2}}{3})$-approximation algorithms for the problem of minimizing the weighted completion time of coflows. In the coflow scheduling problem with release dates, they claimed a deterministic $\frac{67}{3}$-approximation and a randomized $(9+\frac{16\sqrt{2}}{3})$-approximation algorithms. However, Ahmadi \textit{et al.}~\cite{ahmadi2020scheduling} proved that their technique actually yields only a deterministic $\frac{76}{3}$-approximation algorithm for coflow scheduling with release times. Khuller \textit{et al.}~\cite{khuller2016brief} also developed an approximation algorithm for coflow scheduling with arbitrary release times with a ratio of $12$. Moreover, when all coflows had release dates equal to zero, they obtained a deterministic 8-approximation and a randomized $3+2\sqrt{2} \approx 5.83$-approximation. In recent work, Shafiee and Ghaderi~\cite{shafiee2018improved} obtained a 5-approximation algorithm with arbitrary release times, and a 4-approximation algorithm without release time. Ahmadi \textit{et al.}~\cite{ahmadi2020scheduling} proposed a primal-dual algorithm to improve running time for the coflow scheduling problem. Huang \textit{et al.}~\cite{Huang2020} considered scheduling a single coflow on a heterogeneous parallel network and proposed a $O(m)$-approximation algorithm, where $m$ is the number of network cores.

\subsection{Our Contributions}
This paper considers the coflow scheduling problem in identical parallel networks. In the divisible coflow scheduling problem, we first propose a $(6-\frac{2}{m})$-approximation algorithm with arbitrary release times, and a $(5-\frac{2}{m})$-approximation without release time, where $m$ is the number of network cores. When coflow is indivisible, we propose a $(4m+1)$-approximation algorithm with arbitrary release times, and a $(4m)$-approximation without release time.

\subsection{Organization}
The rest of this article is organized as follows. Section \ref{sec:Preliminaries} introduces basic notations and preliminaries. Section \ref{sec:Algorithm1} presents an algorithm for divisible coflow scheduling. Section \ref{sec:Algorithm2} presents an algorithm for indivisible coflow scheduling. Section~\ref{sec:Results} compares the performance of the previous algorithms with that of the proposed algorithm. Section \ref{sec:Conclusion} draws conclusions.

\section{Notation and Preliminaries}\label{sec:Preliminaries}
This work abstracts the identical parallel networks as a set $M$ of $m$ giant $N \times N$ non-blocking switch, with $N$ input links connected to $N$ source servers and $N$ output links connected to $N$ destination servers. Each switch represents a network core. This abstract model is simple and practical, as topological designs such as Fat-tree or Clos~\cite{al2008scalable} enable the construction of networks with fully bisection bandwidth. In $N$ source servers (or destination servers), the $i$-th source server (or $j$-th destination server) is connected to the $i$-th input (or $j$th output) port of each parallel switch. Therefore, each source server (or destination server) has $m$ simultaneous uplinks (or downlinks). Each uplink (or downlink) can be a bundle of multiple physical links in the actual topology~\cite{Huang2020}. Let $\mathcal{I}$ be the source server set and $\mathcal{J}$ be the destination server set. The network core can be viewed as a bipartite graph, with $\mathcal{I}$ on one side and $\mathcal{J}$ on the other side. For simplicity, we assume that all network cores are the same and all links in each network core have the same capacity (or the same speed).

A coflow consists of a set of independent flows whose completion time is determined by the completion time of the latest flow in the set. The coflow $k$ can be expressed as $N\times N$ demand matrix $D^{(k)}=\left(d_{ijk}\right)_{i,j=1}^{N}$ where $d_{ijk}$ denote the size of the flow to be transferred from input $i$ to output $j$ in coflow $k$. In other words, each flow is a triple $(i, j, k)$, where $i \in \mathcal{I}$ is its source node and $j \in \mathcal{J}$ is its destination node, $k$ is the coflow to which it belongs. Moreover, we assume that flows consist of discrete data units, so their sizes are integers. For simplicity, we assume that all flows in a coflow arrive at the system at the same time (as shown in~\cite{Qiu2015}).

This work considers the following offline coflow scheduling problem with release dates. There is a set of coflows denoted by $\mathcal{K}$. Coflow $k$ is released to the system at time $r_k$, $k = 1, 2, \ldots, |\mathcal{K}|$, which means it can only be scheduled after time $r_k$. Let $C_k$ be the completion time of coflow $k$, that is, the time for all flows of coflow $k$ to finish processing. Each coflow $k\in \mathcal{K}$ has a positive weight $w_{k}$. Weights can capture different priorities for different coflows. The higher the weight, the higher the priority. The goal is to schedule coflows in an identical parallel network to minimize $\sum_{k\in \mathcal{K}} w_{k}C_{k}$, the total weighted completion time of the coflow. If all weights are equal, the problem is equivalent to minimizing the average coflow completion time. Table~\ref{tab:notations} presents the notation and terminology that are used herein.
\begin{table}[ht]
\caption{Notation and Terminology}
\vspace{2mm}
    \centering
        \begin{tabular}{||c|p{4in}||}
    \hline
     $m$      & The number of network cores.          \\
    \hline    
     $N$      & The number of input/output ports.         \\
    \hline
     $\mathcal{I}, \mathcal{J}$ & The source server set and the destination server set.         \\
    \hline    
     $\mathcal{K}$ & The set of coflows.         \\
    \hline
     $D^{(k)}$     & The demand matrix of coflow $k$. \\
    \hline    
     $d_{ijk}$     & The size of the flow to be transferred from input $i$ to output $j$ in coflow $k$.   \\
    \hline     
     $C_k$     & The completion time of coflow $k$.   \\
    \hline     
     $C_{ijk}$ & The completion time of flow $(i, j, k)$. In the analysis of approximation algorithm for divisible coflow scheduling, we use $C_{f}$ to represent $C_{ijk}$. \\
    \hline     
     $r_k$     & The released time of coflow $k$.  \\
    \hline     
     $w_{k}$   &  The weight of coflow $k$. \\
    \hline     
		 $\bar{C}_1, \ldots, \bar{C}_n$ & An optimal solution to the linear program. \\
		\hline 
		$\tilde{C}_{1}, \ldots, \tilde{C}_{n}$ & The schedule solution to our algorithm. \\
		\hline 
        \end{tabular}
    \label{tab:notations}
\end{table}


\section{Approximation Algorithm for Divisible Coflow Scheduling}\label{sec:Algorithm1}
In this section, coflows are considered divisible, where different flows in a coflow can be transmitted through different cores. We assume that divisible coflows are transmitted at the flow level, so that data in a flow are distributed to the same core. Let $\mathcal{K}_{i}=\left\{(k, j)| d_{ijk}>0, \forall k\in \mathcal{K}, \forall j\in \mathcal{J} \right\}$ be the set of flows with source $i$. Let $\mathcal{K}_{j}=\left\{(k, i)| d_{ijk}>0, \forall k\in \mathcal{K}, \forall i\in \mathcal{I} \right\}$ be the set of flows with destination $j$. For any subset $S_{i}\subseteq \mathcal{K}_{i}$ (or $S_{j}\subseteq \mathcal{K}_{j}$), let $d(S_{i})=\sum_{(k, j)\in S_{i}} d_{ijk}$ (or $d(S_{j})=\sum_{(k, i)\in S_{j}} d_{ijk}$) and $d^2(S_{i})=\sum_{(k, j)\in S_{i}} d_{ijk}^{2}$ (or $d^2(S_{j})=\sum_{(k, i)\in S_{j}} d_{ijk}^{2}$). We can formulate our problem as the following linear programming relaxation:
\begin{subequations}\label{coflow:main}
\begin{align}
& \text{min}  && \sum_{k \in \mathcal{K}} w_{k} C_{k}     &   & \tag{\ref{coflow:main}} \\
& \text{s.t.} && C_{k} \geq C_{ijk}, && \forall k\in \mathcal{K}, \forall i\in \mathcal{I}, \forall j\in \mathcal{J} \label{coflow:a} \\
&  && C_{ijk}\geq r_k+d_{ijk}, && \forall k\in \mathcal{K}, \forall i\in \mathcal{I}, \forall j\in \mathcal{J} \label{coflow:b} \\
&  && \sum_{(k, j)\in S_{i}}d_{ijk}C_{ijk}\geq \frac{1}{2m} \left(d(S_{i})^2+d^2(S_{i})\right),&& \forall i\in \mathcal{I}, \forall S_{i}\subseteq \mathcal{K}_{i} \label{coflow:c} \\
&  && \sum_{(k, i)\in S_{j}}d_{ijk}C_{ijk} \geq \frac{1}{2m} \left(d(S_{j})^2+d^2(S_{j})\right),&& \forall j\in \mathcal{J}, \forall S_{j}\subseteq \mathcal{K}_{j} \label{coflow:d} 
\end{align}
\end{subequations}

In the linear program (\ref{coflow:main}), $C_{k}$ is the completion time of coflow $k$ in the schedule and $C_{ijk}$ is the completion time of flow $(i, j, k)$. The constraint~(\ref{coflow:a}) is that the completion time of coflow $k$ is bounded by all its flows. The constraint~(\ref{coflow:b}) ensures that the completion time of any flow $(i, j, k)$ is at least its release time $r_k$ plus its load. The constraints~(\ref{coflow:c}) and (\ref{coflow:d}) are used to lower bound the completion time variable in the input port and the output port respectively. These two constraints are modified from the scheduling literature~\cite{Leslie1997}. For any flow $(i, j, k)$, we abstract the indices $i$, $j$ and $k$ and replace them with only one index. Based on the analysis in \cite{Leslie1997}, we have the following two lemmas: 

\begin{lem}\label{lem:lem1}
For the $i$-th input with $n$ flows, let $C_{1}, \ldots, C_{n}$ satisfy (\ref{coflow:c}) and assume without loss of generality that $C_{1}\leq \cdots\leq C_{n}$. Then, for each $f=1,\ldots, n$, if $S=\left\{1, \ldots, f\right\}$,
\begin{eqnarray*}
C_{f}\geq \frac{1}{2m} d(S).
\end{eqnarray*}
\end{lem}
\begin{proof}
For clarity of description, we use $C_{1}, \ldots, C_{n}$ to represent $C_{ijk}$ and use $d_{l}$ to represent $d_{ijk}$ in constraint (\ref{coflow:c}).
According to (\ref{coflow:c}) and the fact that $C_{l}\leq C_{f}$, we have
\begin{eqnarray*}
C_{f}d(S)\geq C_f\sum_{l=1}^{f} d_{l}\geq \sum_{l=1}^{f} d_{l}C_l \geq \frac{1}{2m} \left(d(S)^2+d^2(S)\right)\geq \frac{1}{2m} d(S)^2.
\end{eqnarray*}
The following inequality can be obtained:
\begin{eqnarray*}
C_{f}\geq \frac{1}{2m} d(S).
\end{eqnarray*}
\end{proof}

\begin{lem}\label{lem:lem2}
For the $j$-th output with $n$ flows, let $C_{1}, \ldots, C_{n}$ satisfy (\ref{coflow:d}) and assume without loss of generality that $C_{1}\leq \cdots\leq C_{n}$. Then, for each $f=1,\ldots, n$, if $S=\left\{1, \ldots, f\right\}$,
\begin{eqnarray*}
C_{f}\geq \frac{1}{2m} d(S).
\end{eqnarray*}
\end{lem}
\begin{proof}
The proof is similar to that of lemma~\ref{lem:lem1}.
\end{proof}

Our algorithm flow-driven-list-scheduling (described in Algorithm~\ref{Alg1}) is as follows. Given $n$ flows from all coflows in the coflow set $\mathcal{K}$, we first compute an optimal solution $\bar{C}_1, \ldots, \bar{C}_n$ to the linear program (\ref{coflow:main}). Without loss of generality, we assume $\bar{C}_{1}\leq \cdots\leq \bar{C}_{n}$ and schedule the flows iteratively in the order of this list. For each flow $f$, the problem is to consider all flows that congested with $f$ and scheduled before $f$. Then, find a network core $h$ and assign flow $f$ to network core $h$ such that the complete time of $f$ is minimized. 
Lines 5-10 are to find the least loaded network core and assigning flow to it. Lines 11-24 are modified from Shafiee and Ghaderi's algorithm~\cite{shafiee2018improved}.
Therefore, all flows are transmitted in a preemptible manner.

\begin{algorithm}
\caption{flow-driven-list-scheduling}
    \begin{algorithmic}[1]
		    \REQUIRE a vector $\bar{C}\in \mathbb{R}_{\scriptscriptstyle \geq 0}^{n}$ used to decide the order of scheduling
				\STATE let $load_{I}(i,h)$ be the load on the $i$-th input port of the network core $h$
				\STATE let $load_{O}(j,h)$ be the load on the $j$-th output port of the network core $h$
				\STATE let $\mathcal{A}_h$ be the set of flows allocated to network core $h$
				\STATE both $load_{I}$ and $load_{O}$ are initialized to zero and $\mathcal{A}_h=\emptyset$ for all $h\in [1, m]$
				\FOR{every flow $f$ in non-decreasing order of $\bar{C}_f$, breaking ties arbitrarily}
				    \STATE note that the flow $f$ is sent by link $(i, j)$
				    \STATE $h^*=\arg \min_{h\in [1, m]}\left(load_{I}(i,h)+load_{O}(j,h)\right)$
						\STATE $\mathcal{A}_{h^*}=\mathcal{A}_{h^*}\cup \left\{f\right\}$
						\STATE $load_{I}(i,h^*)=load_{I}(i,h^*)+d_f$ and $load_{O}(j,h^*)=load_{O}(j,h^*)+d_f$
				\ENDFOR
				\FOR{each $h\in [1, m]$ do in parallel}
				    \STATE wait until the first coflow is released
						\WHILE{there is some incomplete flow}
                \FOR{every released and incomplete flow $f\in \mathcal{A}_{h}$ in non-decreasing order of $\bar{C}_f$, breaking ties arbitrarily}
								    \STATE note that the flow $f$ is sent by link $(i, j)$
										\IF{the link $(i, j)$ is idle}
										    \STATE schedule flow $f$
										\ENDIF
								\ENDFOR
								\WHILE{no new flow is completed or released}
								    \STATE transmit the flows that get scheduled in line 17 at maximum rate 1.
								\ENDWHILE
						\ENDWHILE
				\ENDFOR
   \end{algorithmic}
\label{Alg1}
\end{algorithm}

\subsection{Analysis}
This section shows that the proposed algorithm achieves an approximation ratio of $6-\frac{2}{m}$ with arbitrary release times, and an approximation ratio of $5-\frac{2}{m}$ without release time. Let $\mathcal{I}_{i}$ be the set of flows that belongs to the $i$-th input and let $\mathcal{J}_{j}$ be the set of flows that belongs to the $j$-th output. For any flow $f$ with $i$-th input and $j$-th output, let 
\begin{eqnarray*}
F(f)=\left\{l|C_l\leq C_{f}, \forall l\in \mathcal{J}_{j}\right\} 
\end{eqnarray*}
be the set of output flows congested with $f$ and scheduled before $f$.
Let 
\begin{eqnarray*}
G(f)=\left\{l|C_l\leq C_{f}, \forall l\in \mathcal{I}_{i}\right\}
\end{eqnarray*}
be the set of input flows congested with $f$ and scheduled before $f$. Note these two sets also include $f$.

\begin{lem}\label{lem:lem3}
Let $\bar{C}_{1}\leq \cdots\leq \bar{C}_{n}$ be an optimal solution to the linear program (\ref{coflow:main}), and let $\tilde{C}_{1}, \ldots, \tilde{C}_{n}$ denote the completion times in the schedule found by flow-driven-list-scheduling. For each $f=1, \ldots, n$,
\begin{eqnarray*}
\tilde{C}_{f}\leq \left(6-\frac{2}{m}\right)\bar{C}_{f}.
\end{eqnarray*}
\end{lem}
\begin{proof}
The proof is similar to Hall \emph{et al}.~\cite{Leslie1997}. Assume the flow $f$ is sent via link $(i, j)$. Let $S_{j} = F(f)$, $S_{i}=G(f)$, $S=S_{i}\cup S_{j}$ and $r_{max}(S)=\max_{k\in S} r_k$. Consider the schedule induced by the flows $S$. Since all links $(i, j)$ in the network cores are busy from $r_{max}(S)$ to the start of flow $f$, we have
\begin{eqnarray}
\tilde{C}_{f} & \leq & r_{max}(S) + \frac{1}{m}d(S_{i}\setminus \left\{f\right\})+ \frac{1}{m}d(S_{j}\setminus \left\{f\right\})+d_f \label{eq:2}\\
              & \leq & \bar{C}_{f} + \frac{1}{m}d(S_{i}\setminus \left\{f\right\})+ \frac{1}{m}d(S_{j}\setminus \left\{f\right\})+d_f \label{eq:3}\\
							& =    & \bar{C}_{f} + \frac{1}{m}d(S_{i})+ \frac{1}{m}d(S_{j})+\left(1-\frac{2}{m}\right)d_f \label{eq:4}\\
							& \overset{\mbox{Lem.~\ref{lem:lem1}}}{\leq} & 3\bar{C}_{f} + \frac{1}{m}d(S_{j})+\left(1-\frac{2}{m}\right)d_f \label{eq:5}\\
							& \overset{\mbox{Lem.~\ref{lem:lem2}}}{\leq} & 5\bar{C}_{f}+\left(1-\frac{2}{m}\right)d_f \label{eq:6}\\
							& \leq & \left(6-\frac{2}{m}\right)\bar{C}_{f} \label{eq:7}
\end{eqnarray}
The inequality (\ref{eq:3}) is due to $\bar{C}_{f}\geq \bar{C}_{k}$ for all $k\in S$, we have $\bar{C}_{f}\geq r_{max}(S)$. The equation (\ref{eq:4}) shifts the partial flow $f$ into the second and third terms. The inequalities (\ref{eq:5}) and (\ref{eq:6}) are based on lemma~\ref{lem:lem1} and lemma~\ref{lem:lem2} respectively. The inequality (\ref{eq:7}) is due to $\bar{C}_{f}\geq d_f$ in the linear program (\ref{coflow:main}).
\end{proof}

According to lemma~\ref{lem:lem3}, we have the following theorem:
\begin{thm}\label{thm:thm1}
The flow-driven-list-scheduling has an approximation ratio of, at most, $6-\frac{2}{m}$.
\end{thm}

When all coflows are release at time zero, we have the following lemma:
\begin{lem}\label{lem:lem4}
Let $\bar{C}_{1}\leq \cdots\leq \bar{C}_{n}$ be an optimal solution to the linear program (\ref{coflow:main}), and let $\tilde{C}_{1}, \ldots, \tilde{C}_{n}$ denote the completion times in the schedule found by flow-driven-list-scheduling. For each $f=1, \ldots, n$,
\begin{eqnarray*}
\tilde{C}_{f}\leq \left(5-\frac{2}{m}\right)\bar{C}_{f}.
\end{eqnarray*}
when all coflows are released at time zero.
\end{lem}
\begin{proof}
Assume the flow $f$ is sent via link $(i, j)$. Let Let $S_{j} = F(f)$, $S_{i}=G(f)$, $S=S_{i}\cup S_{j}$. Consider the schedule induced by the flows $S$. Since all links $(i, j)$ in the network cores are busy from zero to the start of flow $f$, we have
\begin{eqnarray*}
\tilde{C}_{f} & \leq & \frac{1}{m}d(S_{i}\setminus \left\{f\right\})+ \frac{1}{m}d(S_{j}\setminus \left\{f\right\})+d_f \\
							& =    & \frac{1}{m}d(S_{i})+ \frac{1}{m}d(S_{j})+\left(1-\frac{2}{m}\right)d_f \\
							& \leq & 4\bar{C}_{f}+\left(1-\frac{2}{m}\right)d_f \\
							& \leq & \left(5-\frac{2}{m}\right)\bar{C}_{f}
\end{eqnarray*}
\end{proof}

According to lemma~\ref{lem:lem4}, we have the following theorem:
\begin{thm}\label{thm:thm2}
For the special case when all coflows are released at time zero, the flow-driven-list-scheduling has an approximation ratio of, at most, $5-\frac{2}{m}$.
\end{thm}

\section{Approximation Algorithm for Indivisible Coflow Scheduling}\label{sec:Algorithm2}
In this section, coflows are considered to be indivisible, where flows in a coflow can only be transmitted through the same core. For every coflow $k$ and input port $i$, let $L_{ik}=\sum_{j=1}^{N}d_{ijk}$ be the total amount of data that coflow $k$ needs to transmit through input port $i$. Moreover, let $L_{jk}=\sum_{i=1}^{N}d_{ijk}$ be the total amount of data that coflow $k$ needs to transmit through output port $j$. When the coflow is indivisible, we can formulate our problem as the following linear programming relaxation:
\begin{subequations}\label{incoflow:main}
\begin{align}
& \text{min}  && \sum_{k \in \mathcal{K}} w_{k} C_{k}     &   & \tag{\ref{incoflow:main}} \\
& \text{s.t.} && C_{k}\geq r_k+L_{ik}, && \forall k\in \mathcal{K}, \forall i\in \mathcal{I} \label{incoflow:a} \\
&             && C_{k}\geq r_k+L_{jk}, && \forall k\in \mathcal{K}, \forall j\in \mathcal{J} \label{incoflow:b} \\
&  && \sum_{k\in S}L_{ik}C_{k} \geq \frac{1}{2m} \left(\sum_{k\in S} L_{ik}^2+\left(\sum_{k\in S} L_{ik}\right)^2\right),&& \forall i\in \mathcal{I}, \forall S\subseteq \mathcal{K} \label{incoflow:c} \\
&  && \sum_{k\in S}L_{jk}C_{k} \geq \frac{1}{2m} \left(\sum_{k\in S} L_{jk}^2+\left(\sum_{k\in S} L_{jk}\right)^2\right),&& \forall j\in \mathcal{J}, \forall S\subseteq \mathcal{K} \label{incoflow:d} 
\end{align}
\end{subequations}

In the linear program (\ref{incoflow:main}), $C_{k}$ is the completion time of coflow $k$ in the schedule. The constraints~(\ref{incoflow:a}) and (\ref{incoflow:b}) ensure that the completion time of any coflow $k$ is at least its release time $r_k$ plus its load. The constraints~(\ref{incoflow:c}) and (\ref{incoflow:d}) are used to lower bound the completion time variable in the input port and the output port, respectively. These two constraints are modified from the scheduling literature~\cite{Leslie1997, ahmadi2020scheduling}. Based on the analysis in \cite{Leslie1997}, we have the following two lemmas: 

\begin{lem}\label{lem:lem11}
For the $i$-th input with $n$ coflows, let $C_{1}, \ldots, C_{n}$ satisfy (\ref{incoflow:c}) and assume without loss of generality that $C_{1}\leq \cdots\leq C_{n}$. Then, for each $f=1,\ldots, n$, if $S=\left\{1, \ldots, f\right\}$,
\begin{eqnarray*}
C_{f}\geq \frac{1}{2m} \sum_{k\in S} L_{ik}.
\end{eqnarray*}
\end{lem}
\begin{proof}
According to (\ref{incoflow:c}) and the fact that $C_{l}\leq C_{f}$, we have:
\begin{eqnarray*}
C_{f}\sum_{k\in S}L_{ik} \geq \sum_{k\in S}L_{ik}C_k \geq \frac{1}{2m} \left(\sum_{k\in S} L_{ik}^2+\left(\sum_{k\in S} L_{ik}\right)^2\right) \geq \frac{1}{2m} \left(\sum_{k\in S} L_{ik}\right)^2.
\end{eqnarray*}
The following inequality can be obtained:
\begin{eqnarray*}
C_{f}\geq \frac{1}{2m} \sum_{k\in S} L_{ik}.
\end{eqnarray*}
\end{proof}

\begin{lem}\label{lem:lem22}
For the $j$-th output with $n$ coflows, let $C_{1}, \ldots, C_{n}$ satisfy (\ref{incoflow:d}) and assume without loss of generality that $C_{1}\leq \cdots\leq C_{n}$. Then, for each $f=1,\ldots, n$, if $S=\left\{1, \ldots, f\right\}$,
\begin{eqnarray*}
C_{f}\geq \frac{1}{2m} \sum_{k\in S} L_{jk}.
\end{eqnarray*}
\end{lem}
\begin{proof}
The proof is similar to that of lemma~\ref{lem:lem11}.
\end{proof}

Our algorithm coflow-driven-list-scheduling (described in Algorithm~\ref{Alg2}) is as follows. Given a set $\mathcal{K}$ of $n$ coflows, we first compute an optimal solution $\bar{C}_1, \ldots, \bar{C}_n$ to the linear program (\ref{incoflow:main}). Without loss of generality, we assume $\bar{C}_{1}\leq \cdots\leq \bar{C}_{n}$ and schedule all the flows in all coflows iteratively respecting the ordering in this list. For each coflow $k$, we find a network core $h$ that can transmit coflow $k$ such that the complete time of coflow $k$ is minimized. Lines 5-10 are to find a network core $h$ that minimizes the maximum completion time of the coflow $k$.
Lines 10-25 transmit all the flows allocated to the network core $h$ in the order of the completion time of the coflow $\bar{C}$ to which they belong.


\begin{algorithm}
\caption{coflow-driven-list-scheduling}
    \begin{algorithmic}[1]
		    \REQUIRE a vector $\bar{C}\in \mathbb{R}_{\scriptscriptstyle \geq 0}^{n}$ used to decide the order of scheduling
				\STATE let $load_{I}(i,h)$ be the load on the $i$-th input port of the network core $h$
				\STATE let $load_{O}(j,h)$ be the load on the $j$-th output port of the network core $h$
				\STATE let $\mathcal{A}_h$ be the set of coflows allocated to network core $h$
				\STATE both $load_{I}$ and $load_{O}$ are initialized to zero and $\mathcal{A}_h=\emptyset$ for all $h\in [1, m]$
				\FOR{every coflow $k$ in non-decreasing order of $\bar{C}_k$, breaking ties arbitrarily}
				    \STATE $h^*=\arg \min_{h\in [1, m]}\left(\max_{i,j\in [1,N]}\left(load_{I}(i,h)+load_{O}(j,h)+L_{ik}+L_{jk}\right)\right)$
						\STATE $\mathcal{A}_{h^*}=\mathcal{A}_{h^*}\cup \left\{k\right\}$
						\STATE $load_{I}(i,h^*)=load_{I}(i,h^*)+L_{ik}$ and $load_{O}(j,h^*)=load_{O}(j,h^*)+L_{jk}$ for all $i,j\in [1,N]$
				\ENDFOR
				\FOR{each $h\in [1, m]$ do in parallel}
				    \STATE wait until the first coflow is released
						\WHILE{there is some incomplete flow}
						    \STATE for all $k\in \mathcal{A}_{h}$, list the released and incomplete flows respecting the non-decreasing order in $\bar{C}_k$
								\STATE let $L$ be the set of flows in the list
                \FOR{every flow $f\in L$}
								    \STATE note that the flow $f$ is sent by link $(i, j)$
										\IF{the link $(i, j)$ is idle}
										    \STATE schedule flow $f$
										\ENDIF
								\ENDFOR
								\WHILE{no new flow is completed or released}
								    \STATE transmit the flows that get scheduled in line 18 at maximum rate 1.
								\ENDWHILE
						\ENDWHILE
				\ENDFOR
   \end{algorithmic}
\label{Alg2}
\end{algorithm}

\subsection{Analysis}
This section shows that the proposed algorithm achieves an approximation ratio of $4m+1$ with arbitrary release times, and an approximation ratio of $4m$ without release time. 
\begin{lem}\label{lem:lem33}
Let $\bar{C}_{1}\leq \cdots\leq \bar{C}_{n}$ be an optimal solution to the linear program (\ref{coflow:main}), and let $\tilde{C}_{1}, \ldots, \tilde{C}_{n}$ denote the completion times in the schedule found by coflow-driven-list-scheduling. For each $f=1, \ldots, n$,
\begin{eqnarray*}
\tilde{C}_{f}\leq \left(4m+1\right)\bar{C}_{f}.
\end{eqnarray*}
\end{lem}
\begin{proof}
Assume the last completed flow in coflow $f$ is sent via link $(i, j)$. Let $S = \left\{1, \ldots, f\right\}$ and $r_{max}(S)=\max_{k\in S} r_k$. Consider the schedule induced by the coflows $S$. Since all links $(i, j)$ in the network cores are busy from $r_{max}(S)$ to the start of the last completed flow in coflow $f$, we have
\begin{eqnarray*}
\tilde{C}_{f} & \leq & r_{max}(S) + \sum_{k\in S} \left(L_{ik}+L_{jk}\right) \\
              & \leq & \bar{C}_{f} + \sum_{k\in S} \left(L_{ik}+L_{jk}\right) \\
							& \leq & (4m+1)\bar{C}_{f}
\end{eqnarray*}
The second inequality is due to $\bar{C}_{f}\geq \bar{C}_{k}$ for all $k\in S$, we have $\bar{C}_{f}\geq r_{max}(S)$. The last inequality is based on lemma~\ref{lem:lem11} and lemma~\ref{lem:lem22}.
\end{proof}

According to lemma~\ref{lem:lem33}, we have the following theorem:
\begin{thm}\label{thm:thm11}
The coflow-driven-list-scheduling has an approximation ratio of, at most, $4m+1$.
\end{thm}

Similar to the case of divisible coflow, we have the following lemma:
\begin{lem}\label{lem:lem44}
Let $\bar{C}_{1}\leq \cdots\leq \bar{C}_{n}$ be an optimal solution to the linear program (\ref{incoflow:main}), and let $\tilde{C}_{1}, \ldots, \tilde{C}_{n}$ denote the completion times in the schedule found by coflow-driven-list-scheduling. For each $f=1, \ldots, n$,
\begin{eqnarray*}
\tilde{C}_{f}\leq 4m\bar{C}_{f}.
\end{eqnarray*}
when all coflows are released at time zero.
\end{lem}
\begin{proof}
The proof is similar to that of lemmas~\ref{lem:lem4} and \ref{lem:lem33}.
\end{proof}

According to lemma~\ref{lem:lem44}, we have the following theorem:
\begin{thm}\label{thm:thm22}
For the special case when all coflows are released at time zero, the coflow-driven-list-scheduling has an approximation ratio of, at most, $4m$.
\end{thm}

\section{Results and Discussion}\label{sec:Results}
This section compares the approximation ratio of the proposed algorithms to that of the previous algorithms. 
We consider two problems: scheduling a single coflow problem and scheduling mutiple coflow problem.
In the scheduling a single coflow problem, we compares with the algorithm in Huang \textit{et al.}~\cite{Huang2020}.
In the scheduling mutiple coflow problem, we first use Shafiee and Ghaderi's algorithm to obtain the order of coflows.
According to this order, the algorithm in \cite{Huang2020} distributes each coflow to the identical parallel network.
Therefore, this method achieves an approximation ratio of $5m$ with arbitrary release times, and an approximation ratio of $4m$ without release time.

\begin{figure}[!ht]
    \centering
        \includegraphics[width=5.4in]{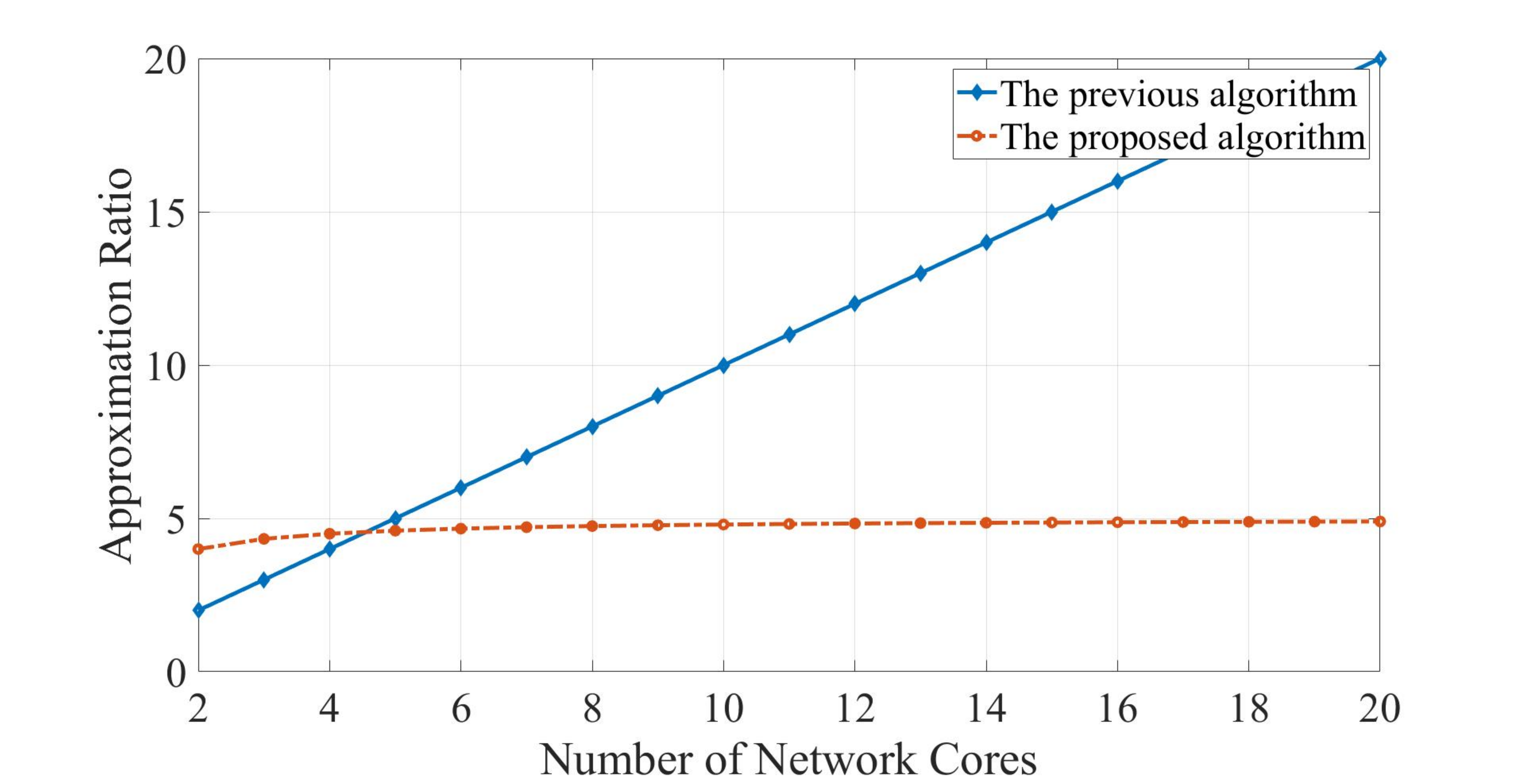}
    \caption{The approximation ratio between the algorithm in \cite{Huang2020} and the proposed algorithm.}
    \label{fig:ratio1}
\end{figure}

\begin{figure}[!ht]
    \centering
        \includegraphics[width=5.4in]{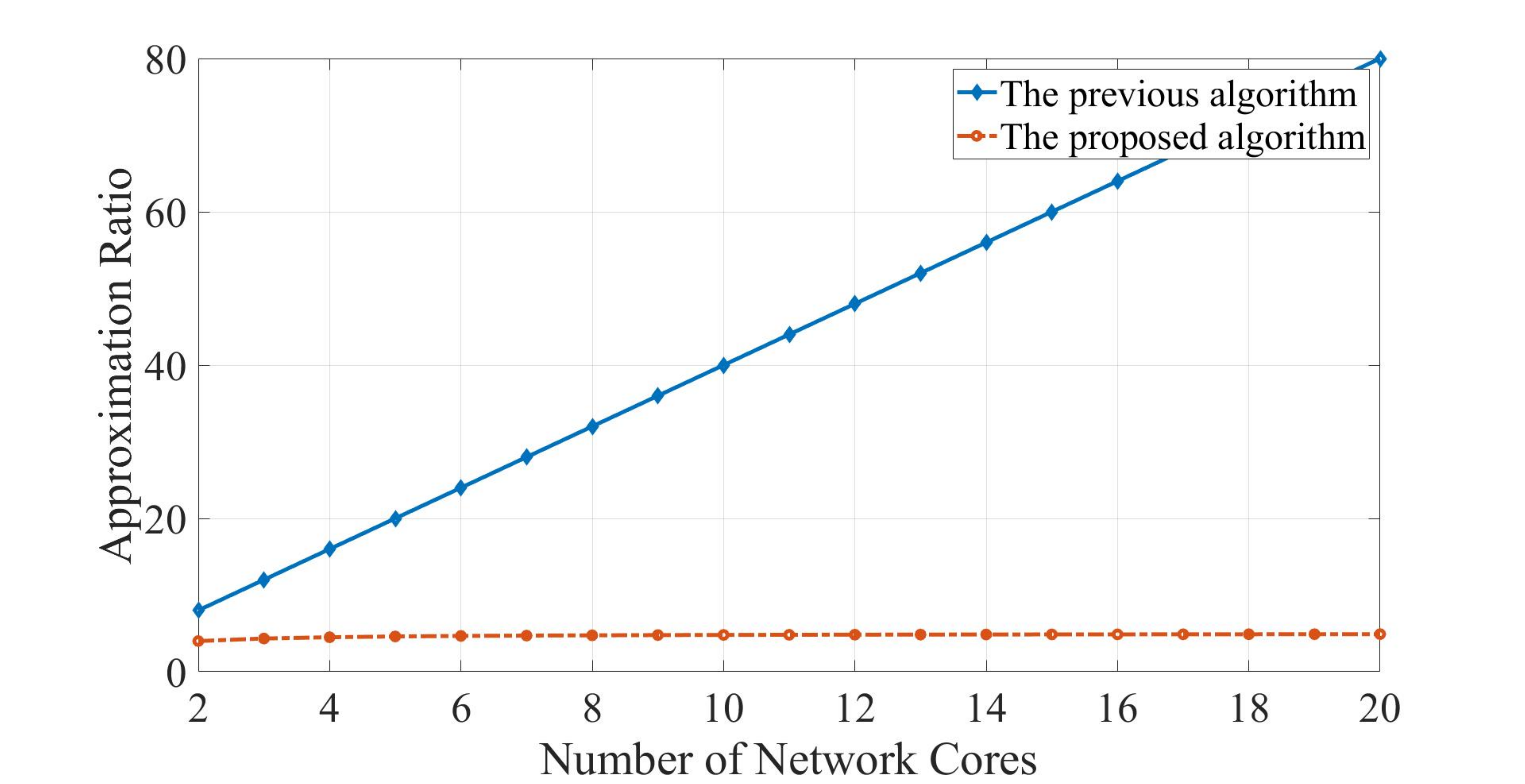}
    \caption{The approximation ratio between the algorithm in \cite{shafiee2018improved, Huang2020} and the proposed algorithm.}
    \label{fig:ratio2}
\end{figure}

\begin{figure}[!ht]
    \centering
        \includegraphics[width=5.4in]{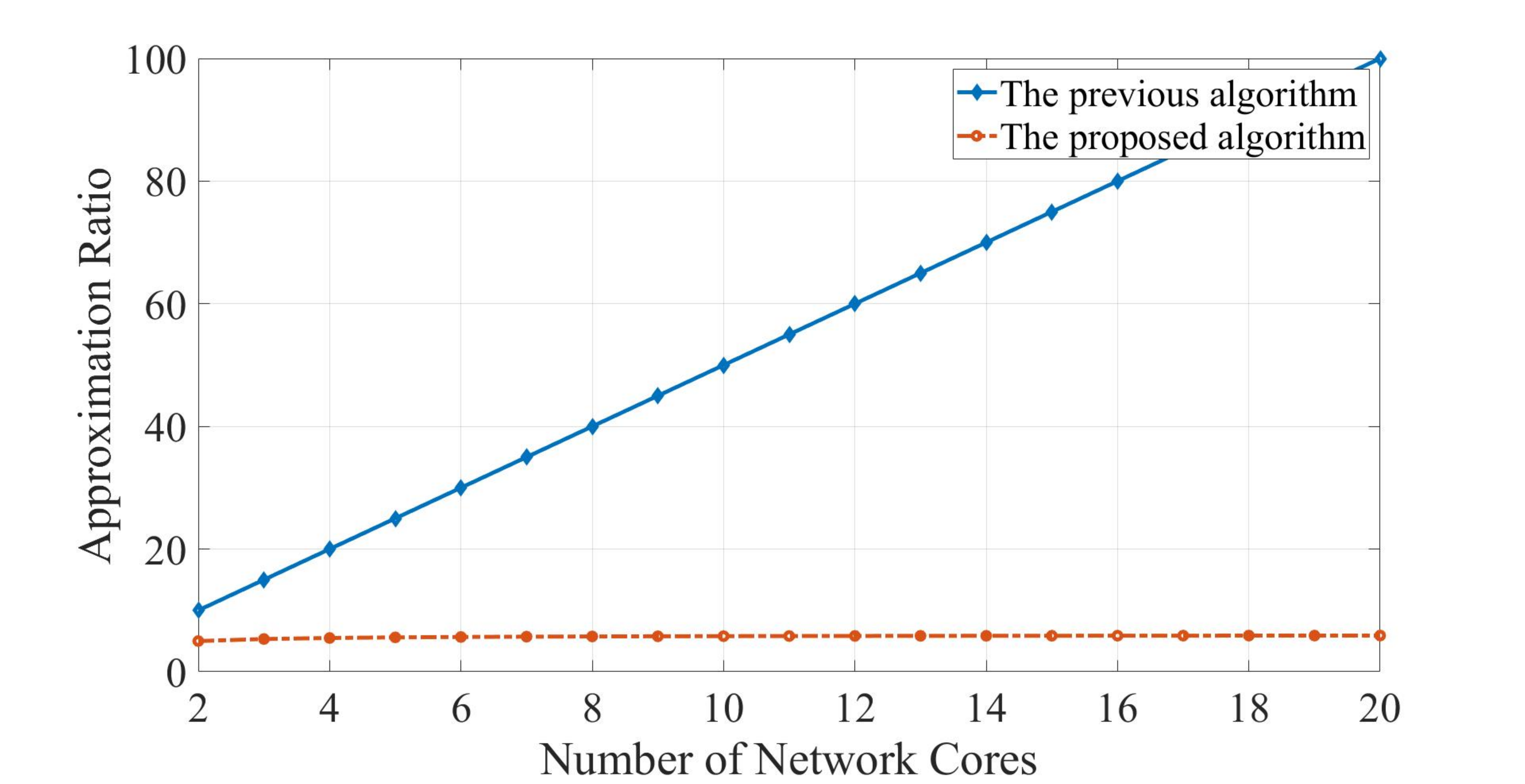}
    \caption{When all coflows are released at time zero, the approximation ratio between the algorithm in \cite{shafiee2018improved, Huang2020} and the proposed algorithm.}
    \label{fig:ratio3}
\end{figure}

In the scheduling a single coflow problem, figure~\ref{fig:ratio1} presents the numerical results concerning the approximation ratio of algorithms.
When $m\rightarrow \infty$, the approximation ratio of proposed algorithm tends to 5.
When $m\geq 5$, the proposed algorithm outperforms the algorithm in \cite{Huang2020}.
Figure~\ref{fig:ratio2} and figure~\ref{fig:ratio3} separately present the numerical results concerning the approximation ratio of algorithms in the scheduling mutiple coflow problem with arbitrary release times and no release times. The proposed algorithm outperforms the previous algorithm in all cases

\section{Concluding Remarks}\label{sec:Conclusion}
With recent technological developments, the single-core model is no longer sufficient. Therefore, we consider the identical parallel network, which is an architecture based on multiple network cores running in parallel. This paper develops two polynomial-time approximation algorithms to solve the coflow scheduling problem in identical parallel networks. Coflow can be considered as divisible or indivisible. Considering the divisible coflow scheduling problem, the proposed algorithm achieves an approximation ratio of $6-\frac{2}{m}$ with arbitrary release times, and an approximation ratio of $5-\frac{2}{m}$ without release time. When coflow is indivisible, the proposed algorithm achieves an approximation ratio of $4m+1$ with arbitrary release times, and an approximation ratio of $4m$ without release time.

\end{document}